\documentclass[oneside,reqno,american]{amsart}
\usepackage[T1]{fontenc}
\usepackage[utf8]{inputenc}
\usepackage{xcolor}
\usepackage{babel}
\usepackage{prettyref}
\usepackage{amstext}
\usepackage{amsthm}
\usepackage{amssymb}
\usepackage[pdfusetitle,
 bookmarks=true,bookmarksnumbered=false,bookmarksopen=false,
 breaklinks=false,pdfborder={0 0 1},backref=false,colorlinks=false]
 {hyperref}
\hypersetup{
 colorlinks=true,citecolor=blue,linkcolor=blue,linktocpage=true}

\makeatletter
\numberwithin{equation}{section}
\numberwithin{figure}{section}

\newrefformat{cor}{Corollary~\ref{#1}}
\newrefformat{subsec}{Section~\ref{#1}}
\newrefformat{lem}{Lemma~\ref{#1}}
\newrefformat{thm}{Theorem~\ref{#1}}
\newrefformat{sec}{Section~\ref{#1}}
\newrefformat{chap}{Chapter~\ref{#1}}
\newrefformat{prop}{Proposition~\ref{#1}}
\newrefformat{exa}{Example~\ref{#1}}
\newrefformat{tab}{Table~\ref{#1}}
\newrefformat{rem}{Remark~\ref{#1}}
\newrefformat{def}{Definition~\ref{#1}}
\newrefformat{fig}{Figure~\ref{#1}}
\newrefformat{claim}{Claim~\ref{#1}}

\makeatother

\theoremstyle{plain}
\newtheorem{thm}{\protect\theoremname}[section]
\newtheorem{lem}[thm]{\protect\lemmaname}
\newtheorem{prop}[thm]{\protect\propositionname}
\newtheorem{cor}[thm]{\protect\corollaryname}
\theoremstyle{remark}
\newtheorem{rem}[thm]{\protect\remarkname}
\newrefformat{cor}{Corollary \ref{#1}}
\newrefformat{lem}{Lemma \ref{#1}}
\newrefformat{prop}{Proposition \ref{#1}}
\newrefformat{rem}{Remark \ref{#1}}
\newrefformat{sec}{Section~\ref{#1}}
\newrefformat{thm}{Theorem \ref{#1}}
\providecommand{\corollaryname}{Corollary}
\providecommand{\lemmaname}{Lemma}
\providecommand{\propositionname}{Proposition}
\providecommand{\remarkname}{Remark}
\providecommand{\theoremname}{Theorem}

\begin{document}
\title{Stability of Maximum-Entropy Inference in Finite Dimensions}
\begin{abstract}
We study maximum-entropy inference for finite-dimensional quantum
states under linear moment constraints. Given expectation values of
finitely many observables, the feasible set of states is convex but
typically non-unique. The maximum-entropy principle selects the Gibbs
state that agrees with the data while remaining maximally unbiased.
We prove that convergence of moments and entropy implies convergence
of states in trace norm, derive explicit quantitative bounds linking
data and entropy deviations to state distance, and show that these
results are stable under unital completely positive maps. The analysis
is self-contained and relies on convex duality, relative entropy,
and Pinsker-type inequalities, providing a rigorous and unified foundation
for finite-dimensional maximum-entropy inference.
\end{abstract}

\author{James Tian}
\address{Mathematical Reviews, 535 W. William St, Suite 210, Ann Arbor, MI
48103, USA}
\email{jft@ams.org}

\maketitle
\tableofcontents{}

\section{Introduction}\label{sec:1}

A central task in quantum information theory is to describe an unknown
quantum state, represented by a density matrix $\rho$ on a finite-dimensional
Hilbert space. Knowing $\rho$ allows one to predict the outcomes
of all measurements and to quantify the resources available for computation,
communication, and metrology. The standard procedure, quantum state
tomography, reconstructs $\rho$ from a complete set of measurement
outcomes. However, the number of measurements needed grows exponentially
with system size \cite{PhysRevLett.105.150401}, making full reconstruction
infeasible for anything beyond small systems. This practical limitation
motivates methods that infer partial information about the state from
a restricted set of measurements.

In many realistic settings, only a small number of observables $X_{1},\dots,X_{k}$
can be measured, giving access to their expectation values
\[
m_{i}=tr\left(\rho X_{i}\right),\quad i=1,\dots,k.
\]
These values define the constraint set
\[
C\left(m\right)=\left\{ \rho'\in\mathcal{S}_{d}:tr\left(\rho'X_{i}\right)=m_{i},\ i=1,\dots,k\right\} ,
\]
the collection of all density matrices consistent with the observed
data. The set $C\left(m\right)$ is convex and compact but typically
large, since the measured data underdetermine the underlying state.
This leads to the inverse problem: among all $\rho'\in C\left(m\right)$,
which one should represent the system?

A principled answer was proposed by E. T. Jaynes, following Shannon's
information-theoretic ideas \cite{MR87305,MR2797301}. The maximum-entropy
principle prescribes choosing the state $\sigma\in C\left(m\right)$
that maximizes the von Neumann entropy $S(\rho)=-tr\left(\rho\log\rho\right).$
The rationale is straightforward: $\sigma$ is consistent with all
known information and assumes nothing beyond it. From a convex-analytic
viewpoint, the maximum-entropy principle is the Legendre dual of the
log-partition function, forming the mathematical backbone of exponential
families \cite{MR274683,MR365798,MR1961186}. In finite dimensions
this unique maximizer has the Gibbs form

\[
\sigma=\frac{\exp\left(-\sum_{i}\lambda_{i}X_{i}\right)}{tr\left(\exp\left(-\sum_{i}\lambda_{i}X_{i}\right)\right)}
\]
where $\lambda=\left(\lambda_{1},\dots,\lambda_{k}\right)$ are real
parameters determined by the constraints. Gibbs states of this form
appear throughout physics. In statistical mechanics they describe
equilibrium ensembles \cite{MR2204925,MR1230389}; in quantum thermodynamics
they act as reference “free” states for resource theories of athermality
\cite{MR3982889,doi:10.1073/pnas.1411728112}; and in data-driven
fields such as quantum machine learning and signal processing they
provide physically meaningful probabilistic models \cite{MR1800071,MR2238070}.
Algorithms for preparing Gibbs states on quantum computers further
demonstrate their operational relevance \cite{MR2570218,MR4538180,MR4207087,MR3676655}. 

Because experimental data are noisy and constraints are rarely exact,
it is important to understand how the maximum-entropy inference behaves
under small perturbations. Three questions are especially relevant:
\begin{enumerate}
\item Stability. If a sequence of states $(\rho_{n})$ satisfies $m(\rho_{n})\to m$
and $S(\rho_{n})\to S(\sigma)$, does this imply that $\rho_{n}\to\sigma$
in trace norm?
\item Quantitative control. Can we bound $\left\Vert \rho-\sigma\right\Vert _{1}$
in terms of the moment mismatch $\left\Vert m(\rho)-m\right\Vert $
and the entropy difference $S(\sigma)-S(\rho)$?
\item Physical consistency. Is this convergence preserved under unital completely
positive (u.c.p.) maps, the natural operations describing physical
post-processing and quantum channels?
\end{enumerate}
Related stability questions have been studied in mathematical physics,
especially in the analysis of ensemble equivalence and thermodynamic
limits \cite{MR793553,MR289084}. There also exist finite-dimensional
investigations of MaxEnt inference, including the continuity properties
of the MaxEnt map \cite{MR3227512,MR3498341} and algorithmic approaches
for Hamiltonian estimation \cite{PhysRevLett.112.190501,MR3036977}.
However, these results address specific aspects in isolation. A unified,
finite-dimensional treatment that establishes existence, uniqueness,
and quantitative stability within a single, information-theoretic
framework appears to be missing. This paper provides such a formulation.

Our analysis is self-contained and based only on convexity, compactness,
and standard properties of quantum entropy. The proofs use an exact
variational identity for the relative entropy together with Pinsker-type
inequalities \cite{MR345574,MR2344161,Watrous_2018}, which relate
relative entropy to the trace distance, and the Arveson extension
theorem \cite{MR247483,MR1976867}, ensuring stability under u.c.p.
maps. The resulting estimates are explicit, dimension-independent,
and directly applicable to quantum information and experimental data
analysis.

\textbf{Organization. }\prettyref{sec:2} introduces the feasible
moment map, characterizes the set of attainable expectation values,
and proves existence and uniqueness of the maximum-entropy state for
any feasible data. \prettyref{sec:3} establishes convergence under
approximate moment and entropy constraints and shows that this convergence
is preserved by u.c.p. post-processing. \prettyref{sec:4} derives
quantitative error bounds, giving explicit rates in terms of entropy
and moment deviations.

\section{States, Constraints, and the Feasible Moment Map}\label{sec:2}

Let $\mathcal{M}=M_{d}\left(\mathbb{C}\right)$ be the algebra of
$d\times d$ complex matrices with the standard trace $tr$. A state
on $\mathcal{M}$ is a positive operator $\rho\geq0$ with $tr\left(\rho\right)=1$.
We treat states both as density matrices and as linear functionals,
via $\rho\left(A\right):=tr\left(\rho A\right)$. Let 
\[
\mathcal{S}_{d}=\left\{ \rho:\rho\geq0,\:tr\left(\rho\right)=1\right\} 
\]
be the set of all density matrices. It is convex (by linearity of
positivity and trace) and compact in finite dimensions (closed and
bounded).

For a state $\rho\in\mathcal{S}_{d}$, define the von Neumann entropy
\[
S\left(\rho\right):=-tr\left(\rho\log\rho\right).
\]
This is finite and continuous in finite dimensions. 

Let $X_{1},\dots,X_{k}\in\mathcal{M}_{sa}$ be fixed self-adjoint
operators (the ``constraints''). Define the unital self-adjoint
subspace
\[
V:=span\left\{ I,X_{1},\dots,X_{k}\right\} \subset\mathcal{M}.
\]
Thus, $V$ is a concrete operator system, but we will use only that
it is a unital $*$-closed linear subspace.

For a state $\rho\in\mathcal{S}_{d}$, we consider its values on $V$,
meaning the linear functional 
\[
A\mapsto tr\left(\rho A\right),\quad A\in V.
\]
Equivalently, we record the scalar moments
\[
m\left(\rho\right):=\left(m_{1}\left(\rho\right),\dots,m_{k}\left(\rho\right)\right)\in\mathbb{R}^{k}
\]
where $m_{i}\left(\rho\right):=tr\left(\rho X_{i}\right)$, $i=1,\dots,k$. 

Fix a target vector $m=\left(m_{1},\dots,m_{k}\right)\in\mathbb{R}^{k}$.
Define the constraint set
\[
C\left(m\right):=\left\{ \rho\in\mathcal{S}_{d}:tr\left(\rho X_{i}\right)=m_{i},\:i=1,\dots,k\right\} .
\]
We will be interested in the case where $C\left(m\right)$ is nonempty.

The following lemma characterizes feasibility.
\begin{lem}
Given $X_{1},\dots,X_{k}\in\mathcal{M}_{sa}$, let
\[
M_{X}:=\left\{ m\left(\rho\right):\rho\in\mathcal{S}_{d}\right\} \subset\mathbb{R}^{k}.
\]
Then $M_{X}$ is a compact convex set, and $C\left(m\right)\neq\emptyset$
if and only if $m\in M_{X}$. Furthermore, the following are equivalent:
\begin{enumerate}
\item \label{enu:b1-1}$m\in M_{X}$ (i.e., $C\left(m\right)\neq\emptyset$)
\item \label{enu:b-1-2}For every $\lambda=\left(\lambda_{1},\cdots,\lambda_{k}\right)\in\mathbb{R}^{k}$,
\[
\sum^{k}_{i=1}\lambda_{i}m_{i}\leq\lambda_{max}\left(\sum\nolimits^{k}_{i=1}\lambda_{i}X_{i}\right),
\]
where $\lambda_{max}\left(A\right)$ denotes the largest eigenvalue
of a self-adjoint matrix $A$. 
\item \label{enu:b-1-3}We have 
\[
M_{X}=conv\left\{ \left(\left\langle \psi,X_{1}\psi\right\rangle ,\dots,\left\langle \psi,X_{k}\psi\right\rangle \right):\psi\in\mathbb{C}^{d},\:\left\Vert \psi\right\Vert =1\right\} .
\]
Equivalently, $m$ is a convex combination of the moment vectors of
pure states.
\end{enumerate}
\end{lem}

\begin{proof}
Recall the set of density matrices $\mathcal{S}_{d}$ is compact,
convex. The map $T:\mathcal{S}_{d}\rightarrow\mathbb{R}^{k}$, 
\[
T\left(\rho\right)=m\left(\rho\right)=\left(tr\left(\rho X_{1}\right),\dots,tr\left(\rho X_{k}\right)\right)
\]
is linear (hence continuous). Therefore, its image $M_{X}=T\left(\mathcal{S}_{d}\right)$
is compact and convex. 

By definition, 
\[
C\left(m\right)\neq\emptyset\Longleftrightarrow\exists\rho\in\mathcal{S}_{d},\ \text{s.t. }T\left(\rho\right)=m\Longleftrightarrow m\in M_{X}.
\]
This proves compactness, convexity, and the equivalence $C\left(m\right)\neq\emptyset\Leftrightarrow m\in M_{X}$. 

For a self-adjoint matrix $H$, we have 
\[
\sup_{\rho\in\mathcal{S}_{d}}tr\left(\rho H\right)=\lambda_{max}\left(H\right).
\]
In fact, for any $\rho\in\mathcal{S}_{d}$, the spectral bound $H\leq\lambda_{max}\left(H\right)I$
yields 
\[
tr\left(\rho H\right)\leq\lambda_{max}\left(H\right)tr\left(\rho\right)=\lambda_{max}\left(H\right).
\]
Equality is achieved by the rank-one projection $\rho=\left|\psi_{max}\left\rangle \right\langle \psi_{max}\right|$
onto a unit eigenvector $\psi_{max}$ for $\lambda_{max}\left(H\right)$. 

Consequently, for $\lambda\in\mathbb{R}^{k}$, writing $H_{\lambda}:=\sum^{k}_{i=1}\lambda_{i}X_{i}$,
\begin{align}
\sup_{y\in M_{X}}\sum\nolimits^{k}_{i=1}\lambda_{i}y_{i} & =\sup_{\rho\in\mathcal{S}_{d}}\sum\nolimits^{k}_{i=1}\lambda_{i}tr\left(\rho X_{i}\right)\nonumber \\
 & =\sup_{\rho\in\mathcal{S}_{d}}tr\left(\rho\sum\nolimits^{k}_{i=1}\lambda_{i}X_{i}\right)\nonumber \\
 & =\lambda_{max}\left(H_{\lambda}\right).\label{eq:b1}
\end{align}
Thus the support function $h_{M_{X}}\left(\lambda\right):=\sup_{y\in M_{X}}\left\langle \lambda,y\right\rangle $
satisfies 
\begin{equation}
h_{M_{X}}\left(\lambda\right)=\lambda_{max}\left(\sum\nolimits^{k}_{i=1}\lambda_{i}X_{i}\right).\label{eq:b2}
\end{equation}

\eqref{enu:b1-1}$\Rightarrow$\eqref{enu:b-1-2}. Assume $m\in M_{X}$.
Then for any $\lambda\in\mathbb{R}^{k}$, using \eqref{eq:b1}-\eqref{eq:b2},
we get
\[
\sum^{k}_{i=1}\lambda_{i}m_{i}\leq\sup_{y\in M_{X}}\sum^{k}_{i=1}\lambda_{i}y_{i}=h_{M_{X}}\left(\lambda\right)=\lambda_{max}\left(H_{\lambda}\right).
\]
Hence \eqref{enu:b-1-2} holds. 

\eqref{enu:b-1-2}$\Rightarrow$\eqref{enu:b1-1}. Assume for every
$\lambda\in\mathbb{R}^{k}$, 
\[
\left\langle \lambda,m\right\rangle \leq\lambda_{max}\left(\sum\nolimits^{k}_{i=1}\lambda_{i}X_{i}\right)=h_{M_{X}}\left(\lambda\right).
\]
Define
\[
K=\left\{ y\in\mathbb{R}^{k}:\left\langle \lambda,y\right\rangle \leq h_{M_{X}}\left(\lambda\right),\;\forall\lambda\in\mathbb{R}^{k}\right\} .
\]
$K$ is the intersection of closed half-spaces, hence a closed convex
set. We always have $M_{X}\subset K$ (since $h_{M_{X}}$ is the support
function of $M_{X}$). Conversely, if $y\notin M_{X}$, the Hahn-Banach
separation theorem provides a $\lambda$ with 
\[
\left\langle \lambda,y\right\rangle >\sup_{z\in M_{X}}\left\langle \lambda,z\right\rangle =h_{M_{X}}\left(\lambda\right),
\]
so $y\notin K$. Hence $K=M_{X}$. Since \eqref{enu:b-1-2} states
exactly that $m\in K$, we conclude that $m\in M_{X}$. Thus \eqref{enu:b1-1}
holds. 

\eqref{enu:b1-1}$\Leftrightarrow$\eqref{enu:b-1-3}. Let $\mathcal{P}:=\left\{ \left|\psi\left\rangle \right\langle \psi\right|:\psi\in\mathbb{C}^{d},\:\left\Vert \psi\right\Vert =1\right\} $
be the set of pure states (rank-1 projections). The extreme points
of $\mathcal{S}_{d}$ are exactly $\mathcal{P}$. Because $\mathcal{S}_{d}=conv\left(\mathcal{P}\right)$
and $T$ is linear, we have 
\begin{align*}
M_{X} & =T\left(\mathcal{S}_{d}\right)=T\left(conv\left(\mathcal{P}\right)\right)=conv\left(T\left(\mathcal{P}\right)\right)\\
 & =conv\left\{ \left(\left\langle \psi,X_{1}\psi\right\rangle ,\dots,\left\langle \psi,X_{k}\psi\right\rangle \right):\left\Vert \psi\right\Vert =1\right\} .
\end{align*}
This is \eqref{enu:b-1-3}. In particular, $m\in M_{X}$ iff $m$
is a convex combination of pure state moment vectors. 
\end{proof}
The following result (\prettyref{prop:b-2}) is standard. Strict concavity
of entropy gives uniqueness; KKT/Lagrange duality yields the exponential
form in the interior; boundary points are limits of interior solutions
via compactness/continuity. See, e.g., \cite{MR87305,MR496300,MR681693,MR1230389},
and \cite{MR274683}. We include a proof for the reader's convenience
and to establish the notation and key concepts used throughout the
paper.
\begin{prop}[maximum entropy under linear constraints]
\label{prop:b-2} Let $X_{1},\dots,X_{k}\in\mathcal{M}_{sa}$ and
$m\in M_{X}$. Then the entropy $S\left(\rho\right)=-tr\left(\rho\log\rho\right)$
has a unique maximizer $\sigma$ on $C\left(m\right)$. If $m\in ri\left(M_{X}\right)$,
the relative interior of $M_{X}$, there exists a unique $\lambda\in\mathbb{R}^{k}$
with 
\[
\sigma=\frac{\exp\left(-\sum^{k}_{i=1}\lambda_{i}X_{i}\right)}{tr\left(\exp\left(-\sum^{k}_{i=1}\lambda_{i}X_{i}\right)\right)},\qquad tr\left(\sigma X_{i}\right)=m_{i},
\]
and $\sigma$ is full rank. If $m\in\partial M_{X}$, the maximizer
is still unique, may be rank-deficient, and arises as a limit of full-rank
Gibbs states associated with $m^{\left(n\right)}\rightarrow m$, $m^{\left(n\right)}\in ri\left(M_{X}\right)$. 
\end{prop}

\begin{proof}
Note the set $C\left(m\right)$ is a closed (linear) slice of the
compact state space and hence compact and convex. The map $\rho\mapsto S\left(\rho\right)$
is continuous and strictly concave on the convex set of states. Therefore,
a maximizer exists and is unique. 

Consider the convex function 
\[
\phi\left(\lambda\right):=\log tr\left(\exp\left(-\sum\nolimits^{k}_{i=1}\lambda_{i}X_{i}\right)\right),\quad\lambda\in\mathbb{R}^{k}.
\]
The function $\phi$ is strictly convex, and its gradient is given
by 
\[
\nabla\phi\left(\lambda\right)=\left(-tr\left(\sigma_{\lambda}X_{1}\right),\dots,-tr\left(\sigma_{\lambda}X_{k}\right)\right),\qquad\sigma_{\lambda}:=\frac{\exp\left(-\sum^{k}_{i=1}\lambda_{i}X_{i}\right)}{tr\left(\exp\left(-\sum^{k}_{i=1}\lambda_{i}X_{i}\right)\right)}.
\]
Define the dual (Legendre-Fenchel transform)
\[
\phi^{*}\left(m\right):=\sup_{\lambda\in\mathbb{R}^{k}}\left\{ -\left\langle \lambda,m\right\rangle -\phi\left(\lambda\right)\right\} .
\]
Then the entropy maximization problem is equivalent to: for any feasible
$\rho$, 
\[
S\left(\rho\right)=\inf_{\lambda}\left\{ \phi\left(\lambda\right)+\left\langle \lambda,m\left(\rho\right)\right\rangle \right\} \leq\phi\left(\lambda\right)+\left\langle \lambda,m\right\rangle ,
\]
with equality exactly when $\rho=\sigma_{\lambda}$ and $tr\left(\sigma_{\lambda}X_{i}\right)=m_{i}$.
Moreover, 
\[
\sup_{\rho\in C\left(m\right)}S\left(\rho\right)=-\inf_{\lambda}\left\{ \phi\left(\lambda\right)+\left\langle \lambda,m\right\rangle \right\} =\phi^{*}\left(m\right).
\]

When $m\in ri\left(M_{X}\right)$, Slater’s condition (existence of
a full-rank feasible $\rho$) holds. Then the dual optimum is attained
at a finite $\lambda$, and first-order optimality yields 
\[
\nabla_{\lambda}\left(\phi\left(\lambda\right)+\left\langle \lambda,m\right\rangle \right)=0\Longleftrightarrow-tr\left(\sigma_{\lambda}X_{i}\right)+m_{i}=0,\quad i=1,\dots,k,
\]
hence $tr\left(\sigma_{\lambda}X_{i}\right)=m_{i}$. Thus the unique
maximizer is $\sigma=\sigma_{\lambda}$ in exponential form. Since
$\exp\left(\cdot\right)$ is strictly positive, $\sigma$ is full
rank. 

(Equivalent “KKT” viewpoint.) Maximizing $S\left(\rho\right)$ over
linear constraints has Lagrangian 
\[
\mathcal{L}\left(\rho,\alpha,\lambda\right)=-tr\left(\rho\log\rho\right)-\alpha\left(tr\left(\rho\right)-1\right)-\sum^{k}_{i=1}\lambda_{i}\left(tr\left(\rho X_{i}\right)-m_{i}\right).
\]
Stationarity in $\rho$ gives $-\log\rho-I-\lambda I-\sum^{k}_{i=1}\lambda_{i}X_{i}=0$,
hence $\rho\propto\exp\left(-\sum_{i}\lambda_{i}X_{i}\right)$. Interior
feasibility ensures finite multipliers and full rank.

If $m\in\partial M_{X}$, existence and uniqueness of a maximizer
still hold as above. The Gibbs form with finite $\lambda$ may fail,
and the maximizer can be rank-deficient. To see it as a limit of full-rank
Gibbs states, pick any sequence $m^{\left(n\right)}\in ri\left(M_{X}\right)$
with $m^{\left(n\right)}\rightarrow m$. For each $n$, using the
interior case above, there is a unique full rank 
\[
\sigma_{n}=\frac{\exp\left(-\sum_{i}\lambda^{\left(n\right)}_{i}X_{i}\right)}{tr\left(\exp\left(-\sum_{i}\lambda^{\left(n\right)}_{i}X_{i}\right)\right)},\qquad tr\left(\sigma_{n}X_{i}\right)=m^{\left(n\right)}_{i}.
\]

The state space is compact, so $\sigma_{n}$ has cluster points. Let
$\sigma$ be any limit point. Since the moment map is continuous,
we get $tr\left(\sigma_{n}X_{i}\right)=m^{\left(n\right)}_{i}\rightarrow m_{i}$,
hence $tr\left(\sigma X_{i}\right)=m_{i}$, so $\sigma\in C\left(m\right)$.
By upper semicontinuity of entropy (indeed continuity in finite dimensions),
\[
S\left(\sigma\right)\geq\limsup_{n\rightarrow\infty}S\left(\sigma_{n}\right)=\sup_{\rho\in C\left(m\right)}S\left(\rho\right),
\]
so $\sigma$ is a maximizer in $C\left(m\right)$. Uniqueness from
above forces all cluster points to coincide, so $\sigma_{n}\rightarrow\sigma$.
Since the constraints pushed to the boundary, $\sigma$ may have reduced
rank (its support sits on the exposed face determined by a supporting
hyperplane of $M_{X}$). 
\end{proof}

\section{Entropy-Constrained Convergence and u.c.p. Stability}\label{sec:3}

This section explains why matching the constraint data is not, by
itself, enough to recover the underlying state, and how entropy closes
that gap. The constraints specify a small collection of observables
and target values. Many states can share those same expectation values.
Among them, one state is singled out by the entropy principle: the
unique state of maximum entropy consistent with the constraints. We
refer to it as the max-entropy state.

Our first goal is a convergence principle. Whenever a sequence of
states reproduces the constraint values more and more accurately and,
at the same time, loses no entropy relative to the max-entropy state,
the sequence must converge to that state. Intuitively, the constraints
pin down the “location” along an affine slice, while entropy selects
the point of least additional structure, and asking for both to hold
in the limit leaves no room for the sequence to drift elsewhere.

A second goal is operational stability. The constraints generate a
natural observable subspace. Any post-processing built from that subspace
(compressions, coarse-grainings, or measurement pipelines implemented
by unital completely positive maps) should preserve the convergence.
In practice, this means that once convergence holds at the level of
states, every downstream statistic built from the constraints will
also converge.

The results below formalize these ideas. \prettyref{thm:c1} gives
a clean sufficient condition for convergence to the max-entropy state.
Then \prettyref{cor:c2} shows that this convergence is stable under
all unital completely positive post-processings on the constraint
subspace. A brief remark records the easy converse of \prettyref{thm:c1}
in finite dimensions, and together these pieces yield a full equivalence.
\begin{thm}
\label{thm:c1}Let $X_{1},\dots,X_{k}\in M_{d}\left(\mathbb{C}\right)_{sa}$
and let $m\in M_{X}$. Let $\sigma\in C\left(m\right)$ be the unique
entropy maximizer from \prettyref{prop:b-2}. Suppose $\left(\rho_{n}\right)\subset\mathcal{S}_{d}$
satisfies:
\begin{enumerate}
\item \label{enu:c1-1}$tr\left(\rho_{n}X_{i}\right)\rightarrow m_{i}$,
$i=1,\dots,k$;
\item \label{enu:c1-2}$S\left(\rho_{n}\right)\rightarrow S\left(\sigma\right)$.
\end{enumerate}
Then 
\[
\left\Vert \rho_{n}-\sigma\right\Vert _{1}\longrightarrow0,
\]
hence, for every $A\in M_{d}\left(\mathbb{C}\right)$, 
\[
tr\left(\left(\rho_{n}-\sigma\right)A\right)\longrightarrow0.
\]
In particular, 
\[
\sup_{A\in V,\left\Vert A\right\Vert \leq1}\left|tr\left(\left(\rho_{n}-\sigma\right)A\right)\right|\longrightarrow0,
\]
where $V=span\left\{ I,X_{1},\dots,X_{k}\right\} $. 
\end{thm}

\begin{proof}
Suppose $m\in ri\left(M_{X}\right)$. By \prettyref{prop:b-2}, there
is a unique $\lambda\in\mathbb{R}^{k}$ such that 
\[
\sigma=\frac{\exp\left(-\sum^{k}_{i=1}\lambda_{i}X_{i}\right)}{tr\left(\exp\left(-\sum^{k}_{i=1}\lambda_{i}X_{i}\right)\right)},\qquad tr\left(\sigma X_{i}\right)=m_{i}.
\]
For every state $\rho\in\mathcal{S}_{d}$ one has the exact variational
identity 
\begin{equation}
D\left(\rho\Vert\sigma\right)=S\left(\sigma\right)-S\left(\rho\right)+\sum^{k}_{i=1}\lambda_{i}\left(tr\left(\rho X_{i}\right)-m_{i}\right),\label{eq:c1}
\end{equation}
where $D\left(\rho\Vert\sigma\right)=tr\left(\rho\left(\log\rho-\log\sigma\right)\right)$
is the relative entropy. (This follows from a direct calculation.
For details, also see \prettyref{lem:exact-identity}). 

Apply \eqref{eq:c1} to $\rho=\rho_{n}$. By assumptions \eqref{enu:c1-1}-\eqref{enu:c1-2},
\[
S\left(\rho\right)-S\left(\rho_{n}\right)\longrightarrow0,\quad\sum^{k}_{i=1}\lambda_{i}\left(tr\left(\rho_{n}X_{i}\right)-m_{i}\right)\longrightarrow0,
\]
hence $D\left(\rho_{n}\Vert\sigma\right)\rightarrow0$. Quantum Pinsker
inequality with natural logs (see e.g., \cite[Thm 5.38]{Watrous_2018}),
\[
\frac{1}{2}\left\Vert \rho-\sigma\right\Vert ^{2}_{1}\leq D\left(\rho\Vert\sigma\right),
\]
now gives $\left\Vert \rho_{n}-\sigma\right\Vert _{1}\rightarrow0$.
The convergence of $tr\left(\left(\rho_{n}-\sigma\right)A\right)$
for each fixed $A$ follows by Holder's inequality: 
\[
\left|tr\left(\left(\rho_{n}-\sigma\right)A\right)\right|\leq\left\Vert \rho_{n}-\sigma\right\Vert _{1}\left\Vert A\right\Vert .
\]
This proves the theorem in the interior case. 

Next, assume $m\in\partial M_{X}$. Pick any sequence $m^{\left(j\right)}\in ri\left(M_{X}\right)$
with $m^{\left(j\right)}\rightarrow m$. Let $\sigma^{\left(j\right)}$
be the unique full-rank Gibbs maximizer in $C\left(m^{\left(j\right)}\right)$,
so that 
\[
\sigma^{\left(j\right)}=\frac{\exp\left(-\sum^{k}_{i=1}\lambda^{\left(i\right)}_{j}X_{i}\right)}{tr\left(\exp\left(-\sum^{k}_{i=1}\lambda^{\left(j\right)}_{j}X_{i}\right)\right)},\qquad tr\left(\sigma^{\left(j\right)}X_{i}\right)=m^{\left(j\right)}_{i}.
\]
By \prettyref{prop:b-2}, $\sigma^{\left(j\right)}\rightarrow\sigma$
and $S\left(\sigma^{\left(j\right)}\right)\rightarrow S\left(\sigma\right)$. 

Fix $j$. Apply the interior case identity \eqref{eq:c1} with $\sigma^{\left(j\right)}$
and its multipliers $\lambda^{\left(j\right)}$ to $\rho=\rho_{n}$,
we get 
\begin{equation}
D\left(\rho_{n}\Vert\sigma^{\left(j\right)}\right)=S(\sigma^{\left(j\right)})-S\left(\rho_{n}\right)+\sum^{k}_{i=1}\lambda^{\left(j\right)}_{i}\left(tr\left(\rho_{n}X_{i}\right)-m^{\left(j\right)}_{i}\right).\label{eq:c2}
\end{equation}
Letting $n\rightarrow\infty$, using $S\left(\rho_{n}\right)\rightarrow S\left(\sigma\right)$
and $tr\left(\rho_{n}X_{i}\right)\rightarrow m_{i}$, we obtain 
\begin{equation}
\limsup_{n\rightarrow\infty}D\left(\rho_{n}\Vert\sigma^{\left(j\right)}\right)\leq S(\sigma^{\left(j\right)})-S\left(\sigma\right)+\sum^{k}_{i=1}\lambda^{\left(j\right)}_{i}\left(m_{i}-m^{\left(j\right)}_{i}\right).\label{eq:c3}
\end{equation}
Now let $j\rightarrow\infty$. Since $S\left(\sigma^{\left(j\right)}\right)\rightarrow S\left(\sigma\right)$
and $m^{\left(j\right)}\rightarrow m$, the right-hand side of \eqref{eq:c3}
tends to 0. Therefore, 
\begin{equation}
\lim_{j\rightarrow\infty}\limsup_{n\rightarrow\infty}D\left(\rho_{n}\Vert\sigma^{\left(j\right)}\right)=0.\label{eq:c4}
\end{equation}

Finally, use lower semicontinuity of relative entropy in the second
argument and the convergence $\sigma^{\left(j\right)}\rightarrow\sigma$,
we have 
\[
D\left(\rho_{n}\Vert\sigma\right)\leq\liminf_{n\rightarrow\infty}D\left(\rho_{n}\Vert\sigma^{\left(j\right)}\right).
\]
Taking $\limsup_{n\rightarrow\infty}$ on both sides and combining
with \eqref{eq:c4} gives 
\[
\limsup_{n\rightarrow\infty}D\left(\rho_{n}\Vert\sigma\right)\leq\lim_{j\rightarrow\infty}\limsup_{n\rightarrow\infty}D\left(\rho_{n}\Vert\sigma^{\left(j\right)}\right)=0,
\]
so $D\left(\rho_{n}\Vert\sigma\right)\rightarrow0$. Quantum Pinsker
now gives $\left\Vert \rho_{n}-\sigma\right\Vert _{1}\rightarrow0$,
and the linear-observable convergence follows as in the interior case. 
\end{proof}
\prettyref{thm:c1} leaves us with two immediate payoffs. First, it
upgrades convergence of the constraint data and entropy into convergence
of the full state, so every observable stabilizes in the limit. Second,
the conclusion is robust under post-processing in the sense that any
unital completely positive map on the constraint subspace carries
convergent inputs to convergent outputs. We state this as a corollary
below. For completeness, we also add a short remark showing the converse
direction is automatic in finite dimensions, so the two conditions
are in fact equivalent.
\begin{cor}[stability under u.c.p. post-processing]
\label{cor:c2} Let $X_{1},\dots,X_{k}\in M_{d}\left(\mathbb{C}\right)_{sa}$,
$V=span\left\{ I,X_{1},\dots,X_{k}\right\} $, and let $\sigma\in C\left(m\right)$
be the max-entropy state. Assume $\left(\rho_{n}\right)\subset\mathcal{S}_{d}$
satisfies the hypotheses of \prettyref{thm:c1}, so that $\left\Vert \rho_{n}-\sigma\right\Vert _{1}\rightarrow0$. 
\begin{enumerate}
\item \label{enu:c-2-1}(Channel stability) For every completely positive
trace-preserving map (CPTP channel) $\Phi:M_{d}\left(\mathbb{C}\right)\rightarrow M_{r}\left(\mathbb{C}\right)$,
\[
\left\Vert \Phi\left(\rho_{n}\right)-\Phi\left(\sigma\right)\right\Vert _{1}\longrightarrow0.
\]
\item \label{enu:c-2-2}(Operator system stability) Let $T:V\rightarrow M_{r}\left(\mathbb{C}\right)$
be any unital completely positive (u.c.p.) map. Then there exists
a CPTP map $\Phi:M_{d}\left(\mathbb{C}\right)\rightarrow M_{r}\left(\mathbb{C}\right)$
such that, for all $B\in M_{r}\left(\mathbb{C}\right)$ and $\rho\in\mathcal{S}_{d}$,
\[
tr\left(\Phi\left(\rho\right)B\right)=tr(\rho\tilde{T}^{*}\left(B\right)),
\]
where $\tilde{T}:M_{d}\left(\mathbb{C}\right)\rightarrow M_{r}\left(\mathbb{C}\right)$
is any u.c.p. extension of $T$ (exists by Arveson's extension theorem).
Consequently, 
\[
\left\Vert \Phi\left(\rho_{n}\right)-\Phi\left(\sigma\right)\right\Vert _{1}\longrightarrow0,
\]
and, in particular, 
\[
\sup_{\left\Vert B\right\Vert \leq1}\left|tr\left(\left(\rho_{n}-\sigma\right)\tilde{T}^{*}\left(B\right)\right)\right|\longrightarrow0.
\]
\end{enumerate}
\end{cor}

\begin{proof}
We use two standard facts:

\textbf{F1}. Let $\Phi$ be CPTP, and let $\Phi^{*}$ be its adjoint,
i.e., $tr\left(\Phi\left(\rho\right)B\right)=tr\left(\rho\Phi^{*}\left(B\right)\right)$.
Then a standard result for unital positive maps (see, e.g., \cite{MR193530,MR1976867})
$\Phi^{*}$ is u.c.p. and 
\[
\left\Vert \Phi^{*}\left(B\right)\right\Vert \leq\left\Vert B\right\Vert ,\quad\forall B.
\]
Using 
\[
\left\Vert X\right\Vert _{1}=\sup_{\left\Vert B\right\Vert \leq1}\left|tr\left(BX\right)\right|,
\]
we get, for any self-adjoint $X$, 
\begin{align*}
\left\Vert \Phi\left(X\right)\right\Vert _{1} & =\sup_{\left\Vert B\right\Vert \leq1}tr\left(B\Phi\left(X\right)\right)=\sup_{\left\Vert B\right\Vert \leq1}tr\left(\Phi^{*}\left(B\right)X\right)\\
 & \leq\sup_{\left\Vert A\right\Vert \leq1}tr\left(AX\right)=\left\Vert X\right\Vert _{1},
\end{align*}
since $\left\{ \Phi^{*}\left(B\right):\left\Vert B\right\Vert \leq1\right\} \subset\left\{ A:\left\Vert A\right\Vert \leq1\right\} $.

\textbf{F2}. If $T:V\rightarrow M_{r}\left(\mathbb{C}\right)$ is
u.c.p. and $V\subset M_{d}\left(\mathbb{C}\right)$ is an operator
system (unital, $*$-closed subspace), then there exists a u.c.p.
extension $\tilde{T}:M_{d}\left(\mathbb{C}\right)\rightarrow M_{r}\left(\mathbb{C}\right)$
(see e.g., \cite{MR247483} and \cite[Theorem 6.2]{MR1976867}). 

Now we prove the assertions. 

\eqref{enu:c-2-1} This follows with $X=\rho_{n}-\sigma$, i.e., 
\[
\left\Vert \Phi\left(\rho\right)-\Phi\left(\sigma\right)\right\Vert _{1}\leq\left\Vert \rho-\sigma\right\Vert _{1}\xrightarrow[n\rightarrow\infty]{}0.
\]

\eqref{enu:c-2-2} By F2, extend $T$ to u.c.p. $\tilde{T}:M_{d}\left(\mathbb{C}\right)\rightarrow M_{r}\left(\mathbb{C}\right)$.
Define $\Phi$ such that 
\[
tr\left(\Phi\left(\rho\right)B\right):=tr(\rho\tilde{T}^{*}\left(B\right)),\quad\forall\rho\in M_{d}\left(\mathbb{C}\right),\ B\in M_{r}\left(\mathbb{C}\right).
\]
One checks that $\Phi$ is CPTP. So by \eqref{enu:c-2-1}, 
\[
\left\Vert \Phi\left(\rho_{n}\right)-\Phi\left(\sigma\right)\right\Vert _{1}\xrightarrow[n\rightarrow\infty]{}0.
\]
Finally, using the dual form of $\left\Vert \cdot\right\Vert _{1}$
again and the contractivity of $\tilde{T}$ on operator norm, 
\begin{align*}
\sup_{\left\Vert B\right\Vert \le1}\left|tr\left(\left(\rho_{n}-\sigma\right)\tilde{T}^{*}\left(B\right)\right)\right| & =\sup_{\left\Vert B\right\Vert \le1}\left|tr\left(\Phi\left(\rho_{n}-\sigma\right)B\right)\right|\\
 & =\left\Vert \Phi\left(\rho_{n}-\sigma\right)\right\Vert _{1}\xrightarrow[n\rightarrow\infty]{}0,
\end{align*}
as claimed. 
\end{proof}
\begin{rem}
The converse direction of \prettyref{thm:c1} is also true. More precisely,
let $X_{1},\dots,X_{k}\in M_{d}\left(\mathbb{C}\right)_{sa}$ and
$m\in M_{X}$. Let $\sigma\in C\left(m\right)$ be the unique max-entropy
state. For a sequence $\left(\rho_{n}\right)\subset\mathcal{S}_{d}$,
the following are equivalent:
\begin{enumerate}
\item \label{enu:c-3-1}$\left\Vert \rho_{n}-\sigma\right\Vert _{1}\rightarrow0$
\item \label{enu:c-3-2}$m\left(\rho_{n}\right)\rightarrow m$ (i.e., $tr\left(\rho_{n}X_{i}\right)\rightarrow m_{i}$
for all $i$) and $S\left(\rho_{n}\right)\rightarrow S\left(\sigma\right)$.
\end{enumerate}
\end{rem}

\begin{proof}
\eqref{enu:c-3-1}$\Rightarrow$\eqref{enu:c-3-2} For each $i$,
\[
\left|tr\left(\rho_{n}X_{i}\right)-tr\left(\sigma X_{i}\right)\right|=\left|tr\left(\left(\rho_{n}-\sigma\right)X_{i}\right)\right|\leq\left\Vert \rho_{n}-\sigma\right\Vert _{1}\left\Vert X_{i}\right\Vert \xrightarrow[n\rightarrow\infty]{}0.
\]
Hence $m\left(\rho_{n}\right)\rightarrow m$. 

Note that in finite dimensions the von Neumann entropy is continuous
in trace norm. In fact, by the Fannes-Audenaert bound (see e.g., \cite{MR345574,MR2344161,MR3645110,Watrous_2018}),
if $\delta_{n}:=\frac{1}{2}\left\Vert \rho_{n}-\sigma\right\Vert _{1}$,
then 
\[
\left|S\left(\rho_{n}\right)-S\left(\sigma\right)\right|\leq\delta_{n}\log\left(d-1\right)+h_{2}\left(\delta_{n}\right)\xrightarrow[n\rightarrow\infty]{}0,
\]
where $h_{2}$ is the binary entropy. Hence $S\left(\rho_{n}\right)\rightarrow S\left(\sigma\right)$. 

Together with \prettyref{thm:c1}, this yields the full equivalence. 
\end{proof}

\section{Quantitative Rates}\label{sec:4}

This section turns qualitative convergence into numerical error bounds.
The main tool is \eqref{eq:d1}, which acts as an energy that upper-bounds
the trace distance via Pinsker. We first record the exact identity
and then convert it into concrete estimates on $\left\Vert \rho-\sigma\right\Vert _{1}$
and on observable deviations $tr\left(\left(\rho-\sigma\right)A\right)$.
We avoid differential geometric arguments, and all bounds are finite-dimensional
and self-contained.

\textbf{Convention}. Whenever a formula involves multipliers $\lambda$,
we implicitly assume $m\in ri\left(M_{X}\right)$. Statements that
do not use $\lambda$ hold for all $m\in M_{X}$. 
\begin{lem}[Exact identity for the entropy gap]
\label{lem:exact-identity}Assume $m\in M_{X}$, and let $\sigma\in C\left(m\right)$
be the maximizer. If $m\in ri\left(M_{X}\right)$, let $\lambda\in\mathbb{R}^{k}$
be the (unique) multipliers with 
\[
\sigma=\frac{\exp\left(-\sum^{k}_{i=1}\lambda_{i}X_{i}\right)}{tr\left(\exp\left(-\sum^{k}_{i=1}\lambda_{i}X_{i}\right)\right)},\qquad tr\left(\sigma X_{i}\right)=m_{i}.
\]
Then for every state $\rho\in\mathcal{S}_{d}$, we have 
\begin{equation}
D\left(\rho\Vert\sigma\right)=S\left(\sigma\right)-S\left(\rho\right)+\sum^{k}_{i=1}\lambda_{i}\left(tr\left(\rho X_{i}\right)-m_{i}\right).\label{eq:d1}
\end{equation}
In particular, 
\[
D\left(\rho\Vert\sigma\right)=S\left(\sigma\right)-S\left(\rho\right),\quad\forall\rho\in C\left(m\right).
\]
\end{lem}

\begin{proof}
For interior $m$, write $\log\sigma=-\sum^{k}_{i=1}\lambda_{i}X_{i}-\phi\left(\lambda\right)I$
with 
\[
\phi\left(\lambda\right)=\log\left(tr\left(\exp\left(-\sum\nolimits^{k}_{i=1}\lambda_{i}X_{i}\right)\right)\right).
\]
Then 
\begin{align*}
D\left(\rho\Vert\sigma\right) & =tr\left(\rho\log\rho\right)-tr\left(\rho\log\sigma\right)\\
 & =-S\left(\rho\right)+\sum_{i}\lambda_{i}tr\left(\rho X_{i}\right)+\phi\left(\lambda\right),
\end{align*}
while $S\left(\sigma\right)=\phi\left(\lambda\right)+\sum_{i}\lambda_{i}m_{i}$,
giving the identity. 
\end{proof}
The exact identity \eqref{eq:d1} relates three kinds of errors: entropy
gap $S\left(\sigma\right)-S\left(\rho\right)$, moment mismatch $\Delta m=m\left(\rho\right)-m$,
and state error $\left\Vert \rho-\sigma\right\Vert _{1}$. Pinsker’s
inequality turns control of the first two into control of the third
(see \prettyref{prop:pinsker} below). In the exact moment case ($\rho\in C\left(m\right)$),
the linear term vanishes and $\left\Vert \rho-\sigma\right\Vert _{1}$
is governed purely by the entropy gap. In the approximate moment regime
(interior case), the linear correction $\lambda\cdot\Delta m$ quantifies
how sensitive the maximum-entropy projection is to small moment errors.
\begin{prop}[Pinsker rates]
 \label{prop:pinsker} \ 
\begin{enumerate}
\item For any state $\rho\in C\left(m\right)$ and maximizer $\sigma\in C\left(m\right)$,
\begin{equation}
\left\Vert \rho-\sigma\right\Vert _{1}\leq\sqrt{2\left(S\left(\sigma\right)-S\left(\rho\right)\right)}.\label{eq:d-1}
\end{equation}
\item Assume $m\in ri\left(M_{X}\right)$ so that $\sigma$ has finite multipliers
$\lambda\in\mathbb{R}^{k}$ as in \prettyref{lem:exact-identity}.
For any state $\rho\in\mathcal{S}_{d}$ with $\Delta m:=m\left(\rho\right)-m$,
\begin{align}
\left\Vert \rho-\sigma\right\Vert _{1} & \leq\sqrt{2\left|S\left(\sigma\right)-S\left(\rho\right)\right|}+\sqrt{2\left\Vert \lambda\right\Vert \left\Vert \Delta m\right\Vert }.\label{eq:d3}
\end{align}
\end{enumerate}
\end{prop}

\begin{proof}
Take $\rho\in C\left(m\right)$. By \prettyref{lem:exact-identity},
$D\left(\rho\Vert\sigma\right)=S\left(\sigma\right)-S\left(\rho\right)$.
Apply Pinsker's inequality, we get 
\[
\left\Vert \rho-\sigma\right\Vert _{1}\leq\sqrt{2D\left(\rho\Vert\sigma\right)}=\sqrt{2\left(S\left(\sigma\right)-S\left(\rho\right)\right)},
\]
which is \eqref{eq:d-1}. 

Assume $m\in ri\left(M_{X}\right)$. Using \prettyref{lem:exact-identity},
\[
D\left(\rho\Vert\sigma\right)=S\left(\sigma\right)-S\left(\rho\right)+\lambda\cdot\Delta m.
\]
By Pinsker, 
\begin{align*}
\left\Vert \rho-\sigma\right\Vert _{1} & \leq\sqrt{2D\left(\rho\Vert\sigma\right)}=\sqrt{2\left(S\left(\sigma\right)-S\left(\rho\right)+\lambda\cdot\Delta m\right)}\\
 & \leq\sqrt{2\left(\left|S\left(\sigma\right)-S\left(\rho\right)\right|+\left|\lambda\cdot\Delta m\right|\right)}\\
 & \leq\sqrt{2\left(\left|S\left(\sigma\right)-S\left(\rho\right)\right|\right)}+\sqrt{2\left\Vert \lambda\right\Vert \cdot\left\Vert \Delta m\right\Vert }.
\end{align*}
\end{proof}
\begin{cor}[Observable rates on $V$]
\label{cor:V-rate}Let $V=span\left\{ I,X_{1},\dots,X_{k}\right\} $.
For any $A\in V$ with $\left\Vert A\right\Vert \leq1$, 
\begin{equation}
\left|tr\left(\left(\rho-\sigma\right)A\right)\right|\leq\left\Vert \rho-\sigma\right\Vert _{1}\leq\sqrt{2\left(S\left(\sigma\right)-S\left(\rho\right)+\lambda\cdot\Delta m\right)}.\label{eq:d4}
\end{equation}
In particular, if $\rho\in C\left(m\right)$, 
\[
\sup_{A\in V,\left\Vert A\right\Vert \leq1}\left|tr\left(\left(\rho-\sigma\right)A\right)\right|\leq\sqrt{2\left(S\left(\sigma\right)-S\left(\rho\right)\right)}.
\]
\end{cor}

\begin{proof}
Apply \prettyref{prop:pinsker} and Holder's inequality $\left|tr\left(XA\right)\right|\leq\left\Vert X\right\Vert _{1}\left\Vert A\right\Vert $. 
\end{proof}
\begin{rem}
Recall that 
\[
\left|tr\left(\left(\rho-\sigma\right)A\right)\right|\leq\left\Vert \rho-\sigma\right\Vert _{1}\left\Vert A\right\Vert 
\]
i.e., trace distance controls all expectation values. Specializing
to $A\in V=span\left\{ I,X_{1},\dots,X_{k}\right\} $ provides operational
bounds on the very observables that define the constraints. Thus \prettyref{cor:V-rate}
translates the Pinsker rates into uniform control of all constrained
expectations. If one prefers to avoid the interior assumption here,
\eqref{eq:d4} can be written using $\sqrt{2D\left(\rho\Vert\sigma\right)}$
on the right-hand side. When $m\left(\rho\right)=m$, this reduces
to $\sqrt{2\left(S\left(\sigma\right)-S\left(\rho\right)\right)}$. 
\end{rem}

\bibliographystyle{amsalpha}
\bibliography{ref}

\end{document}